\documentclass[11pt]{article}

\usepackage{amssymb}
\usepackage{amsthm}
\usepackage{enumerate}
\usepackage{graphicx}

\newcommand{\spec}{{\rm spec}}
\newcommand{\real}{{\mathbb R}}
\newcommand{\complex}{{\mathbb C}}
\newcommand{\torus}{{\mathbb T}}
\newcommand{\integer}{{\mathbb Z}}
\newcommand{\quat}{{\mathbb H}}

\newcommand{\dual}[1]{{\widehat{#1}}}
\newcommand{\alg}{{\cal A}}
\newcommand{\cb}{{\cal B}}
\newcommand{\cC}{{\cal C}}
\newcommand{\ce}{{\cal E}}
\newcommand{\ci}{{\cal I}}
\newcommand{\hilb}{{\cal H}}
\newcommand{\cT}{{\cal T}}
\newcommand{\cU}{{\cal U}}
\newcommand{\cpt}{{\cal K}}
\newcommand{\ind}{{\rm ind}}
\newcommand{\aut}{{\rm Aut}}
\newcommand{\wt}[1]{\widetilde{#1}}
\newcommand{\wh}[1]{\widehat{#1}}

\newcommand\rmap[1]{\stackrel{#1}\longrightarrow}

\newcommand\lmap[1]{\stackrel{#1}\longleftarrow}
\newcommand\lumap[1]{\vcenter{\llap{\scriptsize$#1$}}\uparrow}
\newcommand\rdmap[1]{\downarrow\vcenter{\rlap{\scriptsize$#1$}}}
\newcommand\ldmap[1]{\vcenter{\llap{\scriptsize$#1$}}\downarrow}
\newcommand\lnemap[1]{\vcenter{\llap{\scriptsize$#1$}}\nearrow}
\newcommand\rnemap[1]{\nearrow\vcenter{\rlap{\scriptsize$#1$}}}
\newcommand\lsemap[1]{\vcenter{\llap{\scriptsize$#1$}}\searrow}

\newtheorem{theorem}{Theorem}
\newtheorem{corollary}[theorem]{Corollary}
\newtheorem{definition}[theorem]{Definition}

\usepackage{geometry}
\begin{document}


\title{T-duality and the bulk-boundary correspondence}

\author{Keith C. Hannabuss\footnote{\tiny{Mathematical Institute, 
AWB, Radcliffe Observatory Quarter,
Woodstock Road, 
Oxford, OX2 6GG, U.K.} {\em email}: kch@balliol.ox.ac.uk}
}

\date{}

\maketitle

\begin{abstract}
String-theoretic T-duality can be exploited to simplify some features  of the bulk -boundary correspondence in condensed  matter theory.
This paper surveys how T-duality links position and momentum space pictures of that correspondence.
\end{abstract}



\section*{Introduction}

This paper reviews how T-duality, borrowed from string theory, simplifies some features  of the bulk -boundary correspondence in condensed  matter theory.
This application of T-duals was first suggested and exploited by Mathai and Thiang \cite{MT1,MT3,MT4}, and was later extended  in joint work with the author \cite{HMT16,HMT17}, and for details of the applications we refer to those papers.

After a  short account of the historical context for this work  in Section 1, Section 2 presents the bulk-boundary correspondence in position space, whilst Section 3 gives the momentum space perspective more suited to electron band theory. In Section 4 Cartier's lattice representation of the canonical commutation relations  motivates the appearance of T-duality. The noncommutative geometric version of this duality is described in more detail in Section 5, and is then shown, in Section 6, to link the position and momentum versions of the bulk-boundary correspondence. Finally there is a brief discussion about whether H-flux,  important in string theoretic T-duality, might also appear in solid state theory. Two Appendices cover relevant background aspects of C$^*$- algebras and noncommutative geometry.

\section{Groups, algebras, and topology in solid state physics}

Symmetries played a crucial role in the rapid evolution of the fledgling solid state theory into its modern form following the 1925--6 discovery of a \lq\lq new quantum theory \rq\rq\ by Heisenberg and Schr\"odinger, \cite{Blo,W27,W28}. Wigner (partly with von Neumann) showed that symmetries preserving transition probabilities must be described by unitary or antiunitary operators providing a projective unitary-antiunitary representation, or, in his terminology, a projective corepresentation $g\mapsto D(g)$ of a symmetry group $G$, with  $D(g)D(h) = \sigma(g,h)D(gh)$, for some $\sigma(g,h) \in \complex$ of modulus 1 \cite{vNW, W59}. Time reversal provides a key example of a symmetry represented antiunitarily.
Further exploiting his insights in the 1950s,  Wigner classified systems into three types \cite{W55,W57,W58}.
Dyson later illuminated this result, by observing that the commuting algebra of an irreducible corepresentation $D$ is, by Schur's Lemma, a real division algebra, and, by Frobenius'  Theorem, this must be $\real$, $\complex$, or $\quat$ (the quaternions),  \cite{D62}. (Wigner's original method  is described in \cite[\S26]{W59}).
Dyson showed further that the same argument applied both to the whole group $G$, and to the normal subgroup $G_u$ of elements $g$ such that $D(g)$ is unitary, and that,  surprisingly, the irreducibles for $G$ and $G_u$ could independently have a commuting algebra $\real$, $\complex$, or $\quat$, so that all nine possibilities:
\begin{eqnarray*}
\real\real, &\real\complex,&\real\quat,\\
\complex\real, &\complex\complex, &\complex\quat,\\
\quat\real,&\quat\complex,&\quat\quat.
\end{eqnarray*}
Finally, he observed that the central $\complex\complex$ possibility could come in two distinct forms, leading to ten classes in all.

It was unclear whether all  ten of Dyson's possibilities could be realised experimentally. However,  there were examples of the three Gaussian matrix ensembles, each also having a chiral version, giving six classes. Then, after three decades of limited activity, Zirnbauer \cite{Z96} with Altland \cite{AZ95,AZ96}  took the idea up again and found four additional Gaussian ensembles based on the study of quantum dots. They also made explicit the precise correspondence between the classes of Gaussian ensembles, and Cartan's classification of symmetric spaces, \cite{Z96,AZ96}, thus hinting at a geometric as well as a group-theoretic classification.
(Three particular classes of  symmetric spaces had appeared in \cite[\S V]{D70}, but with the express regret that \lq\lq a more illuminating\rq\rq\  insight was lacking.)

After that there followed a decade of steadily accelerating activity, with important contributions by numerous authors, culminating Kitaev's synthesis, \cite{K09}, which provided an explicit link to homotopy, K-groups, and symmetric spaces. 
(Surveys doing more justice to the numerous crucial papers are to be found in the review of \cite{HK}, and, from a slightly different perspective, in \cite{SCR}. The reviews by Freed and Moore \cite{FM}, from a topological perspective, and by Prodan and Schulz-Baldes \cite{PS}, from an operator algebraic slant,   also include further developments since Kitaev's paper. Other recent strands are the use of (crystal lattice) equivariant K-theory,  \cite{KS17},
and the topological investigation of Fermi arcs in Weyl semimetals, \cite{MT5,MT6}.)

The summary of the scheme which emerged can be summarised in the following table (named by
Hasan and Kane the \lq\lq Altland--Zirnbauer classification\rq\rq).
\begin{center}
	\begin{tabular}{|c| c c c| c c c c c c c c|}\hline
 Cartan    &$\Theta$ &$\Xi$ &$\Pi$ &1 &2 &3 &4 &5 &6 &7 &8\\ \hline
A  &0 &0 &0  &0 &$\integer$&0 &$\integer$&0 &$\integer$&0 &$\integer$\\
AIII  &0 &0 &0 &$\integer$ &0 &$\integer$&0 &$\integer$&0 &$\integer$ &0\\\hline \hline
AI  &1 &0 &0 &0 &0 &0  &$\integer$ &0 &$\integer_2$ &$\integer_2$ &$\integer$\\
BDI  &1 &1 &1 &$\integer$ &0 &0 &0  &$\integer$ &0 &$\integer_2$ &$\integer_2$\\ 
D &0 &1 &0 &$\integer_2$ &$\integer$ &0 &0 &0  &$\integer$ &0 &$\integer_2$ \\ 
DIII  &-1 &1 &1 &$\integer_2$ &$\integer_2$ &$\integer$ &0 &0 &0  &$\integer$ &0 \\
AII  &-1 &0 &0 &0 &$\integer_2$ &$\integer_2$ &$\integer$ &0 &0 &0  &$\integer$\\
CII  &-1 &-1 &1  &$\integer$ &0 &$\integer_2$ &$\integer_2$ &$\integer$ &0 &0 &0 \\
C  &0 &-1 &0 &0 &$\integer$ &0 &$\integer_2$ &$\integer_2$ &$\integer$ &0 &0 \\
CI  &1 &-1 &1  &0 &0 &$\integer$ &0 &$\integer_2$ &$\integer_2$ &$\integer$ &0\\ \hline
\end{tabular}
\end{center}

\bigskip
The first column in this table gives the Cartan symmetric space corresponding to the physical situation, the next indicates whether time reversal $\Theta$ is absent, 0, or present with $\Theta^2 = \pm 1$, as indicated. The third column indicates particle--hole interchange symmetry with $\Xi^2 = \pm1$, or its absence indicated by 0. The fourth column indicates the presence  or absence of chirality with 1 or 0, respectively. 
It is clear that the first two rows in the table consist simply of repeated copies of the first pair of entries, with these interchanged between the first and second row, whilst the other eight rows  are obtained from the third row by cyclic permutations modulo 8.
The entries $\integer_2$ for presence or absence of e.g. superconducting phase, $\integer$ a quantised  observable such as transverse Hall conductivity in the first row (whilst 0 indicates a default phase of the material.
The periodicities 2 and 8 correspond precisely to those of complex and real K-theories, respectively.

To understand why K-theory could be relevant to condensed matter systems one has only to go back to 
Bloch's pioneering study of periodic systems such as crystals, \cite{Blo}.
Let $V \cong \real^d$ denote the group of spatial translations in $\real^d$, $\wh{V}$ its Pontryagin dual, $L\cong \integer^d$ the subgroup of translations through the crystal lattice, $\wh{L}\cong \torus^d$ its dual, and $L^\perp = \{\xi\in \wh{V}:\xi(\ell) = 1 \forall \ell \in L \}$ the reciprocal lattice.
For the rest of this section we concentrate on $d=3$.
Spatial translations through the crystal lattice must commute with a Hamiltonian $H = P^2/(2m) + \Phi(Q)$ with periodic potential $\Phi$.
The action of the lattice group $L \cong \integer^d$ of translations on the Hilbert space $\hilb\cong L^2(V)$ of the quantum system provides a direct integral decomposition $\hilb = \int^\oplus \hilb(k) dk$ where $T(\ell)$ acts on $\hilb(k)$  as multiplication by $\chi_k(\ell)$ for $\chi_k \in \wh{L}$. (Physicists tend to think of $k$ as an element of $\wh{V}$ which is periodic with respect to the reciprocal lattice $L^\perp$, as follows from the isomorphism $\wh{L} \cong \wh{V}/L^\perp$, and write $\chi_k(\ell) = \exp(ik.\ell)$, where the isomorphic vector groups $\wh{V}$ and $V$ are identified and given the usual Euclidean inner product denoted by a dot.  Sometimes $\wh{L}$ is considered a Brillouin zone, \cite{TKNN}, though the Brillouin zones are usually introduced as subsets of $\wh{V}$  which project  onto $\wh{L} = \wh{V}/L^\perp$, and which, together, cover $\wh{V}$,  up to null sets.  The first Brillouin zone consists of elements of $\wh{V}$ which are closer to 0  than to any other lattice point $\lambda \in L^\perp$  in the Euclidean norm metric. Bellissard's noncommutative Brillouin zone is a crossed product C$^*$-algebra, which allows for disorder, \cite{B,BES}.)

General features can be illuminated within each space $\hilb(k)$, where, as Bloch realised, they become much simpler. In particular, the spectrum of the restriction $H(k)$ of the Hamiltonian to $\hilb(k)$ is discrete, and the ground state energy $E_0(k)$ is nondegenerate, \cite[Vol 1, Ch VI.6,7]{CH}, \cite[Ch XIII.16, Th. XIII.89]{RS}, and a lower bound for the gap $E_1 - E_0$ between the ground state and first excited state energy is known, \cite[Th.2.1]{KS}. It is also known that $\hilb(k)$ depends continuously on $k$, giving a band struture to the energy levels. Since only a finite number of energy levels lie below the Fermi level we can assume that $\hilb(k)$ is finite-dimensional, and $\hilb(k)$  can be regarded as  a finite rank vector bundle over the parameter space. As a theory of certain equivalence classes of vector bundles K-theory is an obvious tool in the study of these systems.

(Kasparov's KK-theory is more powerful and more easily adapts to the use of the real theory rather than ordinary K-theory, \cite{KK88}, and unbounded KK theory has been used to analyse the case of a boundary circle in momentum space \cite{BCR}. However, the familiarity and simplicity of  K-theory is easier to motivate and we shall use it  here.)

\begin{center}
\begin{picture}(250,160)(0,-20)
\put(50,0){\line(1,0){100}}
\put(50,0){\line(3,4){60}}
\put(110,80){\line(1,0){100}}
\put(150,0){\line(3,4){60}}
\put(50,-20){\line(1,0){100}}
\put(150,-20){\line(3,4){60}}
\put(50,-20){\line(0,1){20}}
\put(150,-20){\line(0,1){20}}
\put(210,60){\line(0,1){20}}
\put(90,48){\vector(1,0){90}}
\put(140,55){\makebox(0,0)[l]{transverse}}
\put(135,40){\makebox(0,0)[l]{current}}
\put(118,26){\vector(3,4){35}}
\put(110,15){\makebox(0,0)[b]{applied}}
\put(110,5){\makebox(0,0)[b]{voltage}}
\put(135,48){\vector(0,1){70}}
\put(140,110){\makebox(0,0)[l]{magnetic}}
\put(140,100){\makebox(0,0)[l]{field ${\bf B}$}}
\end{picture}
\end{center}

A notable early application of K-theory to condensed matter physics was stimulated by the 1981 discovery of the Integer Quantum Hall Effect, where, at sufficiently low temperatures, the transverse conductivity of a two-dimensional system of electrons  in a perpendicular magnetic field is quantised in integer multiples of $e^2/\hbar$, \cite{KDP}.  
Theoretical explanations were soon suggested by Laughlin  \cite{L81}, and then by Thouless and collaborators \cite{TKNN}. The latter suggested that the effect was topological: the transverse conductivity was essentially giving the Chern character of a line bundle over the Brillouin zone, \cite{ASS, NTW}.
This was soon refined further when Bellissard  circumvented some conceptual difficulties in the purely geometric approach using Connes' new noncommutative geometry \cite{C85,C94}. Inspired by Kubo's formula, Bellissard modelled the transverse conductivity as a cyclic 2-cocycle, which paired with a K$_0$ class defined by the Fermi projection to give a noncommutative Chern character \cite{B}.  Moreover, the technical condition for the cyclic cocycle to be well-defined was linked to Anderson localisation, \cite{BES}.
This suggests that one might consider a similar C$^*$-algebraic approach here.

A striking feature of the K-theoretic classification is the way in which it links behaviour in different dimensions.
It had already become clear that most of the important action in a plane sample actually comes from the edges, \cite{L81,MMP}. Starting in 2000, Schulz-Baldes, Kellendonk, and Richter modelled this situation in the noncommutative geometry of the integer quantum Hall effect,  proving that the bulk and edge conductivities were the same, \cite{SKR, KRS, KS1, KS2}. More recent developments include \cite{ASV}.

Topological insulators provide another motive for investigating the relationship between bulk and boundary, as the boundary may not share all the symmetries, particularly discrete symmetries,  of the bulk,  and this can lead to surface electrons bridging the gap between the bulk energy bands.  

Accounts of C$^*$-algebra theory and background motivation can be found in many places, such as \cite{RW}, but we have included a brief account of the ideas needed here in the two Appendices.

%
%

\section{The bulk-boundary correspondence}

The bulk, $B$, and boundary, $E$, appearing in condensed matter systems are subsets of  position space  $\real^d$. In quantum theory we want to consider C$^*$-algebras $\cb$ and $\ce$ whose spectra are $B$ and $E$, respectively.
For the applications we use algebras related by simple constructions which capture the naive idea of the bulk as being a sort of cylinder with the boundary as cross section.
This can be made precise using induced algebras. 

\begin{definition}
Let $H$ be a closed subgroup of a separable locally compact group $G$, and let $\alpha$ be a 
homomorphism from  $H$ to the automorphisms $\aut(\alg)$ of a C$^*$-algebra $\alg$.
The induced algebra $\ind_H^G(\alg,\alpha)$
consists of bounded continuous $\alg$-valued functions on $G$, with a periodicity condition:
$$
\ind_H^G(\alg,\alpha) = \{ f\in C_b(G,\alg): f(gh^{-1}) = \alpha(h)[f(g)], g\in G, h\in H\}.
$$
The star operation, addition, and product are defined pointwise.
\end{definition}

We are primarily interested in the case when the boundary has codimension 1, and the boundary algebra $\cb$ is induced from the edge algebra $\ce$. We let $V_t \cong \real$ be a subgroup of translations transverse to the boundary, and suppose that $L_t = L\cap V_t \cong \integer$ has  an action $\alpha$ as automorphisms of $\ce$ so that we may define
$$
\cb = \ind_{L_t}^{V_t}(\ce,\alpha)= \ind_\integer^\real(\ce,\alpha).
$$
(We do not need to be too specific about $\ce$ at this stage, though we shall impose a little more structure later.)
The periodicity constraint means that functions $f \in \cb  = \ind_\integer^\real(\ce,\alpha)$ are determined by their values on the unit interval $[0,1)$. We shall take the projection $\epsilon:\cb \to \ce$ to be given by evaluation at 0: $\epsilon(f) = f(0) \in \ce$.

The interior algebra is then defined to be
$$
\ci = \ker(\epsilon) = \{ f\in C_b(\real,\ce): f(0) = 0\},
$$
with $\iota:\ci \to \cb$ the obvious inclusion map. 
so that we have the standard short exact sequence
$$
0 \longrightarrow \ci \rmap{\iota} \cb \rmap{\epsilon} \ce \longrightarrow 0.
$$
Since $f(\ell)  = \alpha(\ell)[f(0)]$, the functions in $\ci$ vanish at all integers $\ell$, so that we may think of $\ci = C_0((0,1),\ce)$, where $C_0$ consists of continuous functions vanishing at the endpoints. 

From this sequence of C$^*$-algebras  one gets a corresponding exact sequence of K-groups
$$
K_0(\ci) \rmap{K_0(\iota)} K_0(\cb) \rmap{K_0(\epsilon)} K_0(\ce),
$$
where $K_0(\iota)$ and $K_0(\epsilon)$ are the K-theory homomorphisms defined by $\iota$ and $\epsilon$. 
Using their suspensions, we obtain another similar sequence 
$$
K_1(\ci) \rmap{K_1(\iota)} K_1(\cb) \rmap{K_1(\epsilon)} K_1(\ce).
$$
Since $(0,1)$ is homeomorphic to $\real$, we also have $\ci \cong C_0(\real,\ce)$, the suspension, so that $K_j(\ci) \cong K_{j+1}(\ce)$, a fact that we shall exploit later.

Index maps connect the above two sequences, \cite[\S8.3]{Bl98}, \cite[Ch. 8]{WO}, \cite[\S1.3]{CMR}.
To obtain $\delta_1: K_1(\ce) \to K_0(\ci)$ for unital $\cb$ one chooses $u$ in some GL$_m(\cb)$ representing a class in $K_1$, and lifts it to $\wh{u}$ in GL$_{2m}(\cb)$, for example:
$$
\wh{u} = \left(\matrix{ u_+ &u_+u_- -1\cr 1-u_-u_+ &2u_- - u_-u_+u_-\cr}\right),
$$
where $\epsilon_*$ sends $u_\pm$  to $u^{\pm 1}$,
and  $\wh{u} \mapsto u\oplus u^{-1}$. Then, setting $p_m = 1 \oplus 0$, one finds that $\epsilon_*(\wh{u}p_m\wh{u}^{-1} - p_m) = 0$, so that
$\wh{u}p_m\wh{u}^{-1} - p_m$ is lifted from $\ci$. It actually gives a well-defined class in $K_0(\ci)$.
Exploiting the suspensions and periodicity, one also obtains $\delta_0: K_0(\ce) \to K_0(C_0(\real,\ci)) \cong K_1(\ci)$.
(For the sequence from compact to bounded operators and then to the Calkin algebra, the index map can be interpreted as a Fredholm index.)
These all fit into an associated long exact sequence, \cite{WO}:
\medskip
\begin{equation}
\matrix{K_0(\ci) &\rmap{\iota_0} &K_0(\cb) &\rmap{\epsilon_0}  &K_0(\ce) \cr
 \lumap{\delta_1} & & &  &\rdmap{\delta_0}\cr
K_1(\ce) &\lmap{\epsilon_1} &K_1(\cb) &\lmap{\iota_1} &K_1(\ci) \cr}
\end{equation}
where we have introduced the abbreviations $\iota_j = K_j(\iota)$, $\epsilon_j = K_j(\epsilon)$. 
(Were we dealing with the real K-theory we should have twenty-four instead of six terms in the exact  sequence.
Bott periodicity was originally discovered for the homotopy of classical groups, \cite{B59,Mi}, and in some ways it is more naturally discussed in that context, where it can be given direct physical meaning, \cite{K09, SCR, KG, KZ}, but we shall not pursue this further here.)

In general, the two vertical index maps are usually much more difficult to obtain than the horizontal maps which are induced by the geometry,  but, when we identify   $K_j(\ci) \cong K_{j+1}(\ce)$ using the suspension as above, they have the simple form $1-\alpha_*$,  where the map $\alpha_*$ is induced by the action of  $\alpha(1)$, \cite[Ex 9.K]{WO}.

%
%

\section{The bulk and boundary in momentum space}

The above geometric description captures some features of the bulk boundary problem, but the physics really takes  place in the Brillouin zone in momentum space, and this motivated the analysis in \cite{SKR, KRS}.
Instead of starting with the schematic picture 
$$
\hbox{{\rm interior}} \to \hbox{{\rm bulk}} \to \hbox{{\rm boundary}},
$$
Schulz-Baldes, Kellendonk, and Richter thought rather of attaching the boundary to the bulk:
$$
\hbox{{\rm boundary}} \to \hbox{{\rm \lq\lq glue\rq\rq}} \to \hbox{{\rm bulk}}.
$$
The example of the circle $S^1$ bounding the unit disc ${\mathbb D}$ in $\complex$ provides a candidate for the \lq\lq glue\rq\rq.
The Hardy space $H^2(S^1)$ is the subspace of functions in $L^2(S^1)$ which extend holomorphically into the disc, and there is a projection $P: L^2(S^1) \to H^2(S^1)$ (which plays an important role in the theory of positive energy representations of loop groups). Identifying the boundary circle with the complex numbers $e^{i\theta}$ of modulus 1, the projection $P$ sends the basis function $e_n(\theta) = e^{in\theta}$ to the holomorphic function $z^n$ when $n \geq 0$, and kills it if $n <0$.

The projection $P$ is intimately related to the Toeplitz algebra $\cT$  on $\hilb = \ell^2({\mathbb N})$, generated by the shift operator $S$ such that $(Sf)(0) = 0$, and $(Sf)(k) = f(k-1)$, for $k = 1,2,\ldots$.  
We first recall that $L^2(S^1) \cong \ell^2(\integer)$, and, since $P$ kills the negative terms in a Fourier series the Hardy space $H^2(S^1) \cong \ell^2({\mathbb N})$.
The  pointwise multiplication action of $C(S^1)$ on $L^2(S^1) \cong \ell^2(\integer)$ can be restricted to $H^2(S^1) \cong \ell^2({\mathbb N})$ and multiplication by $e^{i\theta}$ takes $e_n$ to $e_{n+1}$, so that it gives the shift operator $S$. Elaborating the details, we discover that the Toeplitz algebra $\cT$ can be identified with $P\, C(S^1)\, P$.

A slightly generalised Toeplitz algebra,  $\wh{\cT} = \cT(\wh{\ce})$, $\cite{CMR}$, provides the glue in the momentum space theory of Kellendonk, Richter and Schulz-Baldes.
We define $\cT(\wh{\ce})$ by requiring that there is an embedding $j: \wh{\ce} \to  \cT(\wh{\ce})$, an action $\alpha$ of the non-negative integers ${\mathbb N}$ on $\wh{\ce}$ and an isometry $J$ such that $J^*j(a)J = j(\alpha(1)[a])$, for all $a\in \wh{\ce}$, \cite[Prop.58]{CMR}.
There is an exact sequence 
$$
0 \longrightarrow \cpt\otimes\wh{\ce} \rmap{\wh{\iota}} \cT(\wh{\ce}) \rmap{\wh{\epsilon}} \wh{\ce}\rtimes \integer \longrightarrow 0,
$$
where $\cpt = \cpt(\hilb)$ is the C$^*$-algebra of compact operators, the closure of the finite rank operators on a Hilbert space $\hilb$.
When one defines $\wh{\cb}^\prime = \wh{\ce}\rtimes \integer$ to be the bulk algebra in momentum space the exact sequence becomes
$$
0 \longrightarrow \cpt\otimes\wh{\ce} \rmap{\wh{\iota}} \cT(\wh{\ce}) \rmap{\wh{\epsilon}} \wh{\cb}^\prime \longrightarrow 0,
$$
As in the geometric picture, with similar abbreviations of notation, there is  a momentum space long exact sequence of  K-groups:
\begin{equation}
\matrix{K_0(\wh{\ce}) &\rmap{\wh{\iota}_0} &K_0(\wh{\cT}) &\rmap{\wh{\epsilon}_0}  &K_0(\wh{\cb}^\prime) \cr
 \uparrow & & &  &\downarrow\cr
K_1(\wh{\cb}^\prime)  &\lmap{\wh{\epsilon}_1} &K_1(\wh{\cT}) &\lmap{\wh{\iota}_1} &K_1(\wh{\ce}), \cr}
\end{equation}
\cite{PV}, in which we have used the Morita equivalence of $\cpt\otimes\wh{\ce}$ and $\wh{\ce}$. 
This time the index maps are more complicated, \cite[8.1]{WO}, \cite[1.3.2]{CMR}, and we shall now show how T-duality enables us to derive this sequence directly from the simpler geometric sequence, whilst suggesting a slightly different, but Morita equivalent, definition of the bulk algebra. (This is already observed from a mathematical perspective by Rieffel, \cite{Ri}, but T-duality provides a direct link with the physics.)

%
%

\section{The Lattice Representation of the CCR}

As pointed out by Weyl \cite[\S IV.14]{W28}, the representations of the canonical commutation relations for the phase space translations $V\times \wh{V} \cong \real^{d}\times\real^d$ correspond to projective unitary representations of $V\times \wh{V}$.
There are unitary representations $U_0(x) = \exp(ix.P/\hbar)$ and $U_1 = \exp(i\xi.X/\hbar)$ of translations of positions and momenta, respectively, where $P=(P_1,P_2,\ldots,P_d)$ and $X=(X_1,\ldots,X_d)$ are momentum and position operators (so that momenta generate spatial translations).
The Heisenberg commutation relations exponentiate to give the Weyl relations:
$$
U_0(x)U_1(\eta) = e^{i(\eta.x)/\hbar}U_1(\eta)U_0(x).
$$
Setting $U(x,\xi) = U_1(\xi)U_0(x)$, and using this, we arrive at
\begin{eqnarray*}
U(x,\xi)U(y,\eta) &=&U_1(\xi) U_0(x)U_1(\eta)U_0(y)\\
&=& e^{i(\eta.x)/\hbar}U_1(\xi)U_1(\eta)U_0(x)U_0(y)\\
&=& e^{i(\eta.x)/\hbar}U_1(\xi+\eta)U_0(x+y)\\
&=& e^{i(\eta.x)/\hbar}U(x+y,\xi+\eta).
\end{eqnarray*}
Thus $U$ is a projective representation of $V\times \wh{V}$ with multiplier $\sigma((x,\xi),(y,\eta)) 
= e^{-i(\eta.x)/\hbar}$.

Though less familiar than the Schr\"odinger representation of the canonical commutation relations on $L^2(\real^d)$ and the Fock representation,  Cartier's lattice representation on $L^2$ sections of a line bundle over $\torus^{2d} \cong (V\times\wh{V})/(L\times L^\perp)$, induced from the lattice $L\times L^\perp$ in phase space also provides a useful realisation of the unique irreducible representation \cite{C66, M70}.

Mackey's normal subgroup analysis \cite{M58}, \cite[\S 12]{M70} shows that $U$ can be realised as the representation induced from $L\times L^\perp$,
$$
(U(y,\eta)\psi)(x,\xi) = \sigma((x,\xi),(y,\eta))^{-1}\psi(x+y,\xi+\eta)
 = e^{i(\eta.x)/\hbar}\psi(x+y,\xi+\eta),
$$ 
on a Hilbert space of (suitably square-integrable) functions $\psi$ satisfying the equivariance condition
$$
\psi(\ell+x,\lambda+\xi) = \sigma((\ell,\lambda),(x,\xi))\psi(x,\xi)
 = e^{-i\xi.\ell/\hbar}\psi(x,\xi),
 $$
for all $(\ell,\lambda)\in L\times L^\perp$, and $(x,\xi) \in V\times \wh{V}$, which is a Bloch type relation for translations through the lattice with quasimomentum $k=-\xi$.
We can identify $\psi(x,-k)$ with the Bloch wave function $\psi_k(x)$. The modified function
$\phi(x,\xi) := \sigma((x,0),(0,\xi))^{-1}\psi(x,\xi)$ 
is $L$-periodic, in the sense that $\phi(x+\ell,\xi) = \phi(x,\xi)$, giving an abstract version of Bloch's factorisation of the wave function into the product of an exponential and a periodic function. 

Differentiating
$$
(U_0(y)\psi)(x,\xi)  = \psi(x+y,\xi), \qquad
(U_1(\eta)\psi)(x,\xi)  = e^{-i(\eta.x)/\hbar}\psi(x,\xi+\eta),
$$ 
with respect to $y$ and $\eta$,  we find that the generators are $P = -i\hbar\partial_x$, and $X = x+i\hbar\partial_\xi$.
(Inducing from a subgroup of index $N$ in $L\times L^\perp$ gives the direct sum of copies of the irreducible, and the representation acts by matrices, which may be useful in other applications.)

For $d=1$ on $\torus^2$ with periods $\{a,\hbar/a\}$ the interchange of position and momentum is given by a Fourier transform ${\cal F}$.
(Actually the periodicity  means that this is not quite a normal Fourier transform, but it can be given explicitly \cite[\S1.4]{LV}.)
Fourier transforms interchange integration and point evalution, and in this context turn complicated index maps into simpler ones.

\begin{center}
\begin{picture}(200,105)(-5,0)
\put(7,12){\vector(1,0){70}} \put(7,12){\vector(0,1){42}} \put(7,14){\vector(0,1){44}} 
\put(77,12){\vector(0,1){42}}
\put(77,14){\vector(0,1){44}}
\put(7,58){\vector(1,0){70}} \put(-7,32){\makebox(0,0){$\hbar a^{-1}$}}
\put(70,3){\makebox(0,0){$a$}}
\put(152,7){\vector(0,1){70}} 
\put(152,7){\vector(1,0){42}} 
\put(154,7){\vector(1,0){44}} 
\put(152,77){\vector(1,0){42}}
\put(154,77){\vector(1,0){44}}
\put(198,7){\vector(0,1){70}} 
\put(172,-7){\makebox(0,0){$\hbar a^{-1}$}}
\put(143,70){\makebox(0,0){$a$}}
\put(95,40){\vector(1,0){42}}
\put(123,50){\makebox(0,0){${\cal F}$}}
\put(143,70){\makebox(0,0){$a$}}
\put(95,40){\vector(1,0){42}}
\put(123,50){\makebox(0,0){${\cal F}$}}
\end{picture}
\end{center}

Diagrams like this motivate the idea of T-duality in string theory, in which context they describe the interchange of string momentum and winding number.
(The idea goes back to the work of Buscher on symmetries of $\sigma$-models, and was quickly taken up by Hull, Townsend and others, \cite{B1, B2,HT}.)

For a more precise view we note that  multiplication by $\rho = \hbar/a^2$ is an isomorphism sending $L = \integer a$ to $L^\perp = \integer\hbar/a^2$, and $V$ to $\wh{V}$.
We now interchange the $L$ and $L^\perp$ periodicities and define
$$
U_\rho(z,\zeta) = e^{-iz.\zeta/\hbar} U(-\rho^{-1}\zeta,\rho z) = e^{-iz.\zeta/\hbar}\exp\left[ \frac{i}{\hbar}\rho z.X\right]\exp\left[ -\frac{i}{\hbar}(\rho^{-1}\zeta).P\right],
$$
so that up to a factor position and momentum have been interchanged. We can now check that 
\begin{eqnarray*}
U_\rho(y,\eta)U_\rho(z,\zeta) &=& e^{-iy.\eta/\hbar}e^{-iz.\zeta/\hbar} U(-\rho^{-1}\eta,\rho y) U(-\rho^{-1}\zeta,\rho z)\\
&=& e^{-iy.\eta/\hbar}e^{-iz.\zeta/\hbar} e^{-i\eta.z/\hbar}U(-\rho^{-1}(\eta+\zeta),\rho (y+z))\\
&=& e^{-iy.\eta/\hbar}e^{-iz.\zeta/\hbar} e^{-i\eta.z/\hbar}e^{i(y+z).(\eta+\zeta)}U_\rho(y+z,\eta+\zeta)\\
&=& e^{-iy.\zeta/\hbar}U_\rho(y+z,\eta+\zeta),
\end{eqnarray*}
so that $U_\rho$ is also a $\sigma$ representation of $V\times \wh{V}$. By the Stone--von Neumann uniqueness theorem $U_\rho$ and $U$ must be unitarily equivalent. (The map $\psi \mapsto \psi_\rho$ given by $\psi_\rho(x,\xi) = e^{-i\xi.x/\hbar}\psi(-\rho^{-1}\xi,\rho x)$ gives the equivalence.)

The suggestion of Mathai and Thiang that T-duality could also be a useful tool in the study of bulk--boundary problems has been explored in a series of papers, \cite{MT1, MT3,MT4,HMT16,HMT17}.
In condensed matter $a$ would be a crystal bond length, perhaps a quarter of a nanometre, $250\times 10^{-12}$m, whereas in string theory it  would be of order the Planck length $10^{-35}$m.

%
%

\section{T-duality in noncommutative geometry}

The algebras $\wh{\cb}$, $\wh{\cT}$ and $\wh{\ce}$ are associated to momentum space, which is obtained by Fourier transforming position space, and  position-momentum interchange was an analogue of string-theoretic momentum winding number interchange implementing the T-duality.
The original T-duality was geometric, interchanging principal torus bundles with a 3-form H-flux over the same base, but it worked smoothly only for circle bundles, or rather restricted classes of flux, \cite{BEM}. However, it was noticed that one could easily describe some more general situations by moving to noncommutative geometry \cite{MR} (or even nonassociative geometry \cite{BHM06,BHM10}). Particularly interesting are the continuous trace algebras (see Appendix B):
Corresponding to the geometric case of a principal circle bundle $P$ with flux $H$ is the continuous trace algebra CT$(P,H)$, and its T-dual is the crossed product CT$(P,H)\rtimes \real$. This is again a continuous trace algebra of the form CT$(\wh{P},\wh{H})$, for another principal circle bundle with flux over the same base $M$, and it is the algebra corresponding precisely with the geometric T-dual. More generally, a principal $V/L$-bundle $P \to M$ with flux $H \in H^3(P,\integer)$ is replaced by the continuous trace algebra $\alg = {\rm CT}(P,H)$. There is an action $\alpha$ of the  vector group $V$  as automorphisms of $\alg$, and the lattice subgroup $L < V$ has trivial action on the spectrum, so that the $V$ action on the spectrum is simply lifted from an action of the torus $V/L$. The T-dual is then the crossed product $\wh{\alg} = \alg\rtimes_\alpha V$ \cite{RR}, which we now define.

\begin{definition}
Let $V$ be an abelian group and $\alpha: V \to \aut(\alg)$ an action as automorphisms of a C$^*$-algebra $\alg$.
The crossed product algebra $\dual{\alg} = \alg\rtimes_\alpha V$ can be identified with functions $f \in C_0(V,\alg)$  equipped with the $\alpha$-twisted convolution product and star operation:
\begin{eqnarray*}
(f_1*f_2)(v) &=& \int_V f_1(u)\alpha_u[f_2(v-u)]\,du,\\
f^*(v) &=& \alpha_v[f(-v)]^*,
\end{eqnarray*}
where $du$ is the Haar measure on $V$.
\end{definition}

The  Pontryagin dual $\wh{V} = {\rm Hom}(V,\torus) \subset C_b(V,\torus)$ acts by multiplication on the T-dual $\alg\rtimes_\alpha V$, that is, for $\xi\in \wh{V} \subset C_b(V)$, we set $\wh{\alpha}(\xi)[f](v) = \xi(v)f(v)$ and then we have the Takai-Takesaki duality Theorem:

\begin{theorem}  
Writing $\wh{\alpha}: \wh{V} \to \aut(\wh{\alg})$ for the action $\wh{\alpha}(\xi)[f](v) = \xi(v)f(v)$ on $\wh{\alg} = \alg\rtimes_\alpha V$, we have
$$
\dual{\alg}\rtimes_{\wh{\alpha}}\dual{V} \cong \alg\otimes \cpt(L^2(V)),
$$
which is Morita equivalent to $\alg$.
\end{theorem}

\begin{proof}
The proof for  C$^*$-algebras can be found in \cite{T74} and for von Neumann algebras in \cite{T67}.
\end{proof}

This theorem tells us that crossed products with abelian groups are  indeed a duality, and confirming that it is a  natural noncommutative analogue of T-duality.

When $\alg = \complex$ with the trivial action of $V$ the crossed product $\complex\rtimes V$  is the C$^*$ group algebra of $V$. The crossed product algebra $\alg\rtimes_\alpha \integer$ and the Toeplitz algebra $\cT(\alg)$ are also closely related and one may show that $\cT(\alg)$ is a subalgebra of $\cT\otimes (\alg\rtimes_\alpha \integer)$, \cite[Prop. 5.8]{CMR}.


%
%

\section{The T-dual pictures of bulk and boundary}

We already have two different ways of looking at the bulk and boundary algebras and $K$-theory:
as the geometric algebras $\cb$ and $\ce$, or as the momentum space algebras $\wh{\cb}^\prime$ and $\wh{\ce}$, with long exact sequences (1) and (2), respectively. 
We shall now show that the momentum space algebras are, up to stable equivalence,  the T-duals of the geometric bulk and boundary algebras $\cb$, $\ce$.
To define the T-dual we need some vector group actions, so let us assume that $V_e \cong\real^{d-1}$ has an action $\varepsilon$ as automorphisms of $\ce$, which commutes with the action $\alpha$ of   $L_t \cong \integer$.
Then we claim that $V \cong \real^d$ has an action $\beta$ on $\cb = \ind_\integer^\real(\ce, \alpha)$, which is the product $\beta = \varepsilon\times\tau$ of the action of $\real^{d-1}$ on $\ce$ and the natural translation action $(\tau(w)f)(v) = f(v-w)$ on an equivariant function $f$ in the induced algebra.

The T-dual bulk and boundary algebras are now defined as the crossed product algebras
$$
\wh{\cb} = \cb\rtimes_\beta V \qquad{\rm and}\qquad \wh{\ce} = \ce\rtimes_\varepsilon V_e.
$$
Since the actions of $\varepsilon$ and $\alpha$ on $\ce$ commute, $\ind_\integer^\real(\ce,\alpha)\rtimes_\varepsilon V_e \cong \ind_\integer^\real(\ce\rtimes_\varepsilon V_e,\wt{\alpha}) = \ind_\integer^\real(\wh{\ce},\wt{\alpha})$, where $\wt{\alpha}(\ell)[f](v) = \alpha(\ell)[f(v)]$.
The definition   $\cb = \ind_\integer^\real(\ce,\wt{\alpha})$ now gives rise to a relationship between $\wh{\cb}$ and $\wh{\ce}$:
\begin{eqnarray*}
\wh{\cb} &:=& \cb\rtimes_\beta V\\ 
&=& \ind_\integer^\real(\ce,\alpha){\rtimes_\beta} (V_e\times \real)\\
&=& \ind_\integer^\real(\wh{\ce},\wt{\alpha})\rtimes_\tau \real. 
\end{eqnarray*}
This is already interesting, but  a dramatic simplification follows from Green's Theorem \cite{G}:

\begin{theorem} For any C$^*$-algebra $\alg$ with an action of a closed subgroup $L$ of a separable locally compact group $V$, and, writing $\tau$ for the natural translation action of $V$ on $\ind_L^V(\alg,\alpha)$, we have an isomorphism $M_\alpha$ 
$$
\ind_L^V(\alg,\alpha)\rtimes_\tau V \rmap{M_\alpha} (\alg\rtimes_\alpha L)\otimes \cpt(L^2(V/L)),
$$
so that $\ind_L^V(\alg,\alpha)\rtimes_\tau V$ is Morita equivalent to $\alg\rtimes_\alpha L$.
\end{theorem}

Applying this to our previous calculation, we therefore are led to the following conclusion: 

\begin{corollary} When $\cb = \ind_\integer^\real(\ce, \alpha)$ the T-duals are related by
$$
\wh{\cb} \cong (\wh{\ce}\rtimes_{\wt{\alpha}} \integer)\otimes \cpt(L^2(\real/\integer)),
$$ 
so that $\wh{\cb}$  is Morita equivalent to $\wh{\cb}^\prime = \wh{\ce}\rtimes_{\wt{\alpha}} \integer$.
\end{corollary}

This shows that the definition of the bulk algebra as a T-dual is Morita equivalent to the definition used in Section 3.

To proceed further and derive the PV sequence linking the K-theory of $\wh{\ce}$ and $\wh{\cb}$ we need Connes' Thom isomorphism theorem \cite{C81}, one of the deep results of the  subject.
Given an action $\alpha$ of $\real^D$ on a C$^*$-algebra $\alg$ we have maps  $K_j:\alg \to K_j(\alg)$ and 
$\kappa_j^D:\alg \to K_{j+D}(\alg\rtimes_\alpha \real^D)$.

\begin{theorem}
The maps  $K_j:\alg \to K_j(\alg)$ and $\kappa_j^D:\alg \to K_{j+D}(\alg\rtimes_\alpha \real^D)$ define the functors from the category of C$^*$ algebras with $\real^D$ actions and $\real^D$-morphisms to abelian groups.
There is a natural equivalence $\phi^j_D$ between these functors (independent of the action $\alpha$), so that for any morphism $f:\alg \to \cC$ between algebras we have a commutative diagram 

$$
\matrix{
K_j(\alg) &\rmap{\phi^j_D(\alg)} &\kappa_j^D(\alg)\cr
\ldmap{K_j(f)} &  &\rdmap{\kappa_j^D(f)}\cr
K_j(\cC) &\rmap{\phi^j_D(\cC)} &\kappa_j^D(\cC).\cr
}
$$
\end{theorem}

To lighten the notation we drop the superscript on $\phi$, and its algebraic argument which will be clear from the context, and just write $\phi_D$.
We also extend our earlier abbreviation to write $f_j = K_j(f)$. The crossed product $\alg\rtimes \real^D$ is generated by compactly supported functions $F: \real^D \to \alg$ which the morphism $f$ maps to $\wt{f}[F] = f\circ F$. This means we can set $\wt{f}_{j+D}= \kappa_j^D(f)$.
We note that the suspension $S\alg = C_0(\alg) \cong \ind_{1}^\real(\alg,1)$, so that Green's map $M_1$ gives the isomorphism  
$$
S\alg\rtimes \real \cong \ind_{1}^\real(\alg,1)\rtimes \real \rmap{M_1} \cpt(L^2(\real))\otimes \alg.
$$ 
So  $S\alg\rtimes \real$ is Morita equivalent to $\alg$, and  $K_{j+1}(\alg) \lmap{(M_1)_*} K_{j+1}(S\alg\rtimes \real) \lmap{\phi_1} K_j(S\alg)$, linking $\phi_1$ to the Bott maps in $K$-theory.

We now have  the basis for our first result on sequences for the T-dual algebras.

 \medskip
\begin{theorem}
The long exact sequences of $K$-groups for the bulk and boundary algebras are related to those of their T-duals by the following diagram:
\begin{center}
\begin{picture}(380,175)(20,40) 
\put(320,70){\makebox(0,0){$K_d(S\wh{\ce})$}}
\put(80,180){\makebox(0,0){$K_{d-1}(S\wh{\ce})$}}
\put(80,70){\makebox(0,0){$K_d(\wh{\ce})$}}
\put(200,70){\makebox(0,0){$K_d(\ind_\integer^\real(\wh{\ce}))$}}
\put(200,42){\makebox(0,0){$K_{d+1}(\wh{\cb})$}}
\put(200,208){\makebox(0,0){$K_d(\wh{\cb})$}}
\put(200,180){\makebox(0,0){$K_{d-1}(\ind_\integer^\real(\wh{\ce}))$}}
\put(320,180){\makebox(0,0){$K_{d-1}(\wh{\ce})$}}
\put(130,100){\makebox(0,0){$K_1(\ce)$}}
\put(200,100){\makebox(0,0){$K_1(\cb)$}}
\put(270,100){\makebox(0,0){$K_1(S\ce)$}}
\put(130,150){\makebox(0,0){$K_0(S\ce)$}}
\put(200,150){\makebox(0,0){$K_0(\cb)$}}
\put(270,150){\makebox(0,0){$K_0(\ce)$}}
\put(103,180){\vector(1,0){60}} 
\put(237,180){\vector(1,0){50}} 
\put(165,70){\vector(-1,0){58}} 
\put(293,70){\vector(-1,0){58}} 
\put(80,80){\vector(0,1){90}} 
\put(320,170){\vector(0,-1){90}}
 \put(149,150){\vector(1,0){35}} 
\put(218,150){\vector(1,0){35}} 
\put(250,100){\vector(-1,0){33}} 
\put(180,100){\vector(-1,0){33}} 
\put(270,140){\vector(0,-1){30}} 
\put(130,110){\vector(0,1){30}} 
\put(110,155){\vector(-1,1){15}}
\put(285,155){\vector(1,1){15}}
\put(110,92){\vector(-1,-1){15}}
\put(285,95){\vector(1,-1){15}}
\put(200,155){\vector(0,1){15}}
\put(200,95){\vector(0,-1){15}}
\put(105,185){\vector(4,1){75}} 
\put(220,205){\vector(4,-1){80}}
\put(175,43){\vector(-4,1){80}} 
\put(298,62){\vector(-4,-1){75}} 
\put(200,190){\vector(0,1){10}} 
\put(200,60){\vector(0,-1){10}} 
\put(145,76){\makebox(0,0){$\wt{\epsilon}_{d}$}}
\put(265,77){\makebox(0,0){$\wt{\iota}_{d}$}}
\put(145,186){\makebox(0,0){$\wt{\iota}_{d-1}$}}
\put(255,186){\makebox(0,0){$\wt{\epsilon}_{d-1}$}}
\put(160,90){\makebox(0,0){$\epsilon_1$}}
\put(240,90){\makebox(0,0){$\iota_1$}}
\put(160,158){\makebox(0,0){$\iota_0$}}
\put(240,158){\makebox(0,0){$\epsilon_0$}}
\put(110,125){\makebox(0,0){$1-\alpha_*$}}
\put(290,125){\makebox(0,0){$1-\alpha_*$}}
\put(60,125){\makebox(0,0){$1-\wt{\alpha}_*$}}
\put(340,125){\makebox(0,0){$1-\wt{\alpha}_*$}}
\put(185,85){\makebox(0,0){$\phi_{d-1}$}}
\put(98,95){\makebox(0,0){$\phi_{d-1}$}}
\put(302,92){\makebox(0,0){$\phi_{d-1}$}}
\put(187,163){\makebox(0,0){$\phi_{d-1}$}}
\put(117,165){\makebox(0,0){$\phi_{d-1}$}}
\put(275,165){\makebox(0,0){$\phi_{d-1}$}}
\put(190,193){\makebox(0,0){$\phi_1$}}
\put(190,55){\makebox(0,0){$\phi_1$}}
\end{picture}
\end{center}

where $\wt{\alpha}_* = \phi_{d-1}\circ\alpha_*\circ\phi_{d-1}^{-1}$, and the outer diagonal maps are defined to make the diagram commute.
\end{theorem}

\bigskip
\begin{proof}
We apply Connes' natural transformation $\phi_{d-1}$ to the long exact sequence (1):
$$
\matrix{
K_0(\ci) &\rmap{\iota_0} &K_0(\cb) &\rmap{\epsilon_0} &K_0(\ce)\cr
\lumap{1-\alpha_*} & & & &\rdmap{1-\alpha_*}\cr
K_1(\ce) &\rmap{\iota_1} &K_1(\cb) &\rmap{\epsilon_1} &K_1(\ci).\cr}
$$
(Here we suppress the identification $K_j(\ci) \cong K_{j+1}(\ce)$, but we shall give the isomorphisms explicitly where that clarifies the role of the above natural transformations.) 
Concentrating on one horizontal row and using the action of  $V_e \cong \real^{d-1}$ on $\ce$ we have 
$$
\matrix{
K_j(\ci) &\rmap{\iota_j} &K_j(\cb) &\rmap{\epsilon_j} &K_j(\ce)\cr
\ldmap{\phi_{d-1}} & &\ldmap{\phi_{d-1}} & &\rdmap{\phi_{d-1}}\cr
K_{j+d-1}(\ci\rtimes V_e) &\rmap{\wt{\iota}_{j+d-1}} &K_{j+d-1}(\cb\rtimes V_e) &\rmap{\wt{\epsilon}_{j+d-1}} &K_{j+d-1}(\ce\rtimes V_e).\cr}
$$

Exploiting our earlier remarks, for $\ci = S\ce$, $\cb  = \ind_\integer^\real(\ce)$, and $\ce\rtimes V_e = \wh{\ce}$, this yields
$$
\matrix{
K_j(S\ce) &\rmap{\iota_j} &K_j(\cb) &\rmap{\epsilon_j} &K_j(\ce)\cr
\ldmap{\phi_{d-1}} & &\ldmap{\phi_{d-1}} & &\rdmap{\phi_{d-1}}\cr
K_{j+d-1}(S\wh{\ce}) &\rmap{\iota_{j+d-1}} &K_{j+d-1}(\ind_\integer^\real(\wh{\ce})) &\rmap{\epsilon_{j+d-1}} &K_{j+d-1}(\wh{\ce}).\cr}
$$

Applying the natural transformation $\phi_{d-1}$ to the original long exact sequence we obtain
\begin{center}
\begin{picture}(380,140)(20,60) 
\put(320,70){\makebox(0,0){$K_d(S\wh{\ce})$.}}
\put(80,180){\makebox(0,0){$K_{d-1}(S\wh{\ce})$}}
\put(200,70){\makebox(0,0){$K_d(\ind_\integer^\real(\wh{\ce}))$}}
\put(80,70){\makebox(0,0){$K_d(\wh{\ce})$}}
\put(200,180){\makebox(0,0){$K_{d-1}(\ind_\integer^\real(\wh{\ce}))$}}
\put(320,180){\makebox(0,0){$K_{d-1}(\wh{\ce})$}}
\put(130,100){\makebox(0,0){$K_1(\ce)$}}
\put(200,100){\makebox(0,0){$K_1(\cb)$}}
\put(270,100){\makebox(0,0){$K_1(S\ce)$}}
\put(130,150){\makebox(0,0){$K_0(S\ce)$}}
\put(200,150){\makebox(0,0){$K_0(\cb)$}}
\put(270,150){\makebox(0,0){$K_0(\ce)$}}
\put(103,180){\vector(1,0){60}} 
\put(237,180){\vector(1,0){50}} 
\put(165,70){\vector(-1,0){58}} 
\put(293,70){\vector(-1,0){58}} 
\put(80,80){\vector(0,1){90}} 
\put(320,170){\vector(0,-1){90}}
 \put(149,150){\vector(1,0){35}} 
\put(218,150){\vector(1,0){35}} 
\put(250,100){\vector(-1,0){33}} 
\put(180,100){\vector(-1,0){33}} 
\put(270,140){\vector(0,-1){30}} 
\put(130,110){\vector(0,1){30}} 
\put(110,155){\vector(-1,1){15}}
\put(285,155){\vector(1,1){15}}
\put(110,95){\vector(-1,-1){15}}
\put(285,95){\vector(1,-1){15}}
\put(200,155){\vector(0,1){15}}
\put(200,95){\vector(0,-1){15}}
\put(145,76){\makebox(0,0){$\wt{\epsilon}_{d}$}}
\put(265,77){\makebox(0,0){$\wt{\iota}_{d}$}}
\put(145,187){\makebox(0,0){$\wt{\iota}_{d-1}$}}
\put(255,186){\makebox(0,0){$\wt{\epsilon}_{d-1}$}}
\put(160,92){\makebox(0,0){$\epsilon_1$}}
\put(240,92){\makebox(0,0){$\iota_1$}}
\put(160,156){\makebox(0,0){$\iota_0$}}
\put(240,156){\makebox(0,0){$\epsilon_0$}}
\put(110,125){\makebox(0,0){$1-\alpha_*$}}
\put(290,125){\makebox(0,0){$1-\alpha_*$}}
\put(60,125){\makebox(0,0){$1-\wt{\alpha}_*$}}
\put(340,125){\makebox(0,0){$1-\wt{\alpha}_*$}}
\put(185,85){\makebox(0,0){$\phi_{d-1}$}}
\put(94,95){\makebox(0,0){$\phi_{d-1}$}}
\put(302,92){\makebox(0,0){$\phi_{d-1}$}}
\put(187,163){\makebox(0,0){$\phi_{d-1}$}}
\put(117,165){\makebox(0,0){$\phi_{d-1}$}}
\put(275,165){\makebox(0,0){$\phi_{d-1}$}}
\end{picture}
\end{center}

Since the maps on the outer rectangle are conjugates of the inner maps the outer rectangle gives another exact sequence.

We can now apply $\phi_1: K_{j-1}(\ind_\integer^\real(\wh{\ce})) \to K_{j}(\ind_\integer^\real(\wh{\ce}\rtimes \real)) = K_{j}(\wh{\cb})$ to expand this diagram to  
that presented in the statement of the theorem, with the outer diagonal maps defined to make the diagram commute.
\end{proof}

\begin{corollary}
The $K$-groups of the T-dual algebras fit into a long exact sequence of the form
\begin{equation}
\matrix{
K_d(\wh{\ce}) &\rmap{\phi_1\circ(1-\wt{\alpha}_*)} &K_d(\wh{\ce}) &\rmap{\phi_1\circ\wt{\iota}_{d-1}\circ\phi_1^{-1}}  &K_d(\wh{\cb})\cr
\lumap{{\wt{\epsilon}_d\circ\phi_1^{-1}}} & & & &\rdmap{\wt{\epsilon}_{d-1}\circ\phi_1^{-1}}\cr
K_{d+1}(\wh{\cb}) &\lmap{\phi_1\circ\wt{\iota}_d\circ\phi_1^{-1}} &K_{d+1}(\wh{\ce}) &\lmap{\phi_1\circ (1-\wt{\alpha}_*)} &K_{d-1}(\wh{\ce}).\cr}
\end{equation}
\end{corollary}

\begin{proof}
The last Theorem's outer hexagon can be rotated and flattened to give the diagram:

$$
\matrix{
K_d(\wh{\ce}) &\rmap{1-\wt{\alpha}_*} &K_{d-1}(S\wh{\ce}) &\rmap{\phi_1\circ\wt{\iota}_{d-1}}  &K_d(\wh{\cb})\cr
\lumap{{\wt{\epsilon}_d\circ\phi_1^{-1}}} & & & &\rdmap{\wt{\epsilon}_{d-1}\circ\phi_1^{-1}}\cr
K_{d+1}(\wh{\cb}) &\lmap{\phi_1\circ\wt{\iota}_d} &K_{d}(S\wh{\ce}) &\lmap{1-\wt{\alpha}_*} &K_{d-1}(\wh{\ce}),\cr}
$$
and, using the isomorphism $(M_1)_*\circ\phi_1: K_{j+1}(\wh{\ce})\to K_j(S\wh{\ce})$ (but suppressing $M_1$  in the diagram), the result follows.
\end{proof}

Since $\wh{\cb}$ and $\wh{\cb}^\prime$ are Morita equivalent, the $K$-groups in this long exact sequence are the same as those appearing in the momentum space bulk-boundary exact sequence.
To investigate whether the maps agree we consider the following diagrams, which  appear in  \cite[Lemma 2.3 proof]{PV} and  (upside down) in Paschke, \cite[p. 483]{Pa}:
$$
\matrix{ 
K_1(\wh{\ce}\otimes\cpt) &\rmap{\wh{\iota}_1} &K_1(\cT(\wh{\ce})) &\rmap{\wh{\epsilon}_1} &K_1(\wh{\ce}\rtimes_{\wt{\alpha}}\integer) &\rmap{\delta} & K_0(\wh{\ce}\otimes \cpt) \cr
\lumap{\cong} &\nwarrow &\lumap{d_*}  &\rnemap{i_*} & & &\cr
K_1(\wh{\ce}) &\rmap{1-\wt{\alpha}_*} &K_1(\wh{\ce}) & &K_1(\wh{\ce}\rtimes_{\wt{\alpha}}\integer)  & & \cr
& &  &\lnemap{i_*} &
&\lsemap{\partial} &\cr
K_1(\wh{\ce}) &\rmap{1- \wt{\alpha}_*} &K_1(\wh{\ce}) & &\lumap{\gamma} 
& &K_0(\wh{\ce})\cr
&  & &\lsemap{\wt{\iota}_0} &
 &\lnemap{\wt{\epsilon}_0} & &\cr
& & & &K_0(\ind_\integer^\real(\wh{\ce})). & &\cr
}
$$
In both components of diagram the algebra $\alg$ has been changed to $\wh{\ce}$.
The top connected diagram uses the maps $d_*$, $i_*$,  and index map $\delta$ of \cite{PV}, but $\alpha^{-1}$ has been changed to $\wt{\alpha}_*$, and $\psi_*$, $\pi_*$ to $\wh{\iota}$ and $\wh{\epsilon}$, to match our conventions. (Such  changes of conventions do not prejudice exactness)
In the lower connected diagram of Paschke, \cite{Pa},  the index map is $\partial$, and we have changed $\rho$ to $\wt{\alpha}$ and made a sign change to match our convention, and also inserted  Paschke's map $\gamma$ whose construction is a key result of the paper.  We have also identified the mapping torus $T_\alpha(\wh{\ce})$ with $\ind_\integer^\real(\wh{\ce})$: the former is just the restriction of the induced algebra functions to $[0,1]$, and the latter the appropriately periodic extension of the functions in the mapping torus. 

The maps $1-\wt{\alpha}_*$ and $i_*$ agree, permitting us to merge the complementary diagrams:
$$
\matrix{ 
K_1(\wh{\ce}\otimes\cpt) &\rmap{\wh{\iota}_1} &K_1(\cT(\wh{\ce})) &\rmap{\wh{\epsilon}_1} &K_1(\wh{\ce}\rtimes_{\wt{\alpha}}\integer) &\rmap{\delta} & K_0(\wh{\ce}\otimes \cpt) \cr
\lumap{\cong} & &\lumap{d_*}  &\lnemap{i_*} &\lumap{(M_{\wt{\alpha})_*}} &\lsemap{\partial} &\rdmap{\cong}\cr
K_1(\wh{\ce}) &\rmap{\phi_1^{-1}\circ(1- \wt{\alpha}_*)} &K_0(S\wh{\ce}) &\longrightarrow &K_1(\wh{\cb})) &\longrightarrow &K_0(\wh{\ce})\cr
&  & &\lsemap{\wt{\iota}_0} &\lumap{\phi_1}
 &\lnemap{\wt{\epsilon}_0} & &\cr
& & & &K_0(\ind_\integer^\real(\wh{\ce})). & &\cr
}
$$

\bigskip
We have made two additional changes: exploiting the factorisation $\gamma = (M_{\wt{\alpha}})_*\circ\phi_1$,  \cite[App. A]{HMT17}, and  replacing $K_1(\wh{\ce})$ by $K_0(S\wh{\ce})$ to emphasize the fact that the lowermost part of the diagram agrees with a section of the long exact sequence for the T-duals appearing above.
The diagram shows how the long exact sequences of T-dual and momentum space $K$-groups are related.
In particular, we see that the index maps $\partial$ and $\delta$ are essentially the same. Considering the right hand side of the diagram, and comparing maps from $K_1(\wh{\cb})$ to $K_0(\wh{\ce})$ we see that $\wt{\epsilon}_0\circ\phi_1^{-1} = \partial\circ (M_{\wt{\alpha}})_*$, showing that, up to the Connes and Green identifications, the index map $\partial$ is given by $\wt{\epsilon}_0$.

\begin{theorem}
The maps between the T-dual $K$-groups in the long exact sequence (3) of  Corollary 8 
agree with those of the momentum space exact sequence, and, in particular, $\wt{\epsilon}_0\circ\phi_1^{-1} = \partial\circ M_{\wt{\alpha}}$ shows that the index map of the momentum space sequence (2) is, up to standard identifications corresponds to the boundary evaluation map $\wt{\epsilon}_0$ of the T-dual theory.
\end{theorem}

\begin{proof}
In outline: we have checked the agreement for one of the index groups, and the other follows using suspensions. The agreement between the remaining maps is more easily checked. By Theorem 7 the index maps in the momentum space sequence  agree with the boundary evaluation maps of the geometric sequence.
\end{proof}

%
%

\section{The H-flux}

The Dixmier--Douady Theorem characterises continuous trace algebras in terms of their spectra and H-flux 3-forms.
In string theory the H-flux has a clear interpretation in terms of brane charges, but 3-forms are not prominent in condensed matter problems.

The role of the principal torus bundle fibres in T-duality is played by the $d$-dimensional torus $V/L$ in the condensed matter systems. There can also be control parameters, which can be regarded as the base of the bundle. Since $H$ is trivial in in one or two dimensions the minimal possibilities would be two control variables for a one dimensional crystal, one for a two-dimensional crystal, or none for a three-dimensional crystal. 

One candidate is the mathematical description of screw dislocations in crystals \cite{RZV, HMT16}

\begin{picture}(400,150)(0,20)
\multiput(40,90)(0,-10){4}{\line(3,-1){110}} 
\multiput(40,90)(10,10){6}{\line(3,-1){110}} 
\multiput(40,90)(12,-4){10}{\line(0,-1){30}} 
\multiput(40,90)(12,-4){10}{\line(1,1){50}} 
\multiput(149,54)(10,10){6}{\line(0,-1){30}}
\multiput(149,54)(0,-10){4}{\line(1,1){50}} 

\put(230,110){\makebox(0,0){dislocate}}
\put(210,100){\vector(1,0){40}} 

\multiput(250,90)(0,-10){4}{\line(3,-1){60}} 
\multiput(250,90)(10,10){4}{\line(3,-1){60}} 
\multiput(280,120)(10,10){3}{\line(3,-1){128}} 
\put(344,98){\line(1,1){20}}
\put(330,90){\line(3,2){14}}
\put(344,98){\line(-5,-6){24}}
\put(330,90){\line(-1,-1){20}} 
\put(330,90){\line(0,-1){9}} 
\put(320,80){\line(0,-1){11}} 
\put(322,70){\line(-1,-1){11}} 
\multiput(310,60)(12,-4){6}{\line(0,-1){30}} 
\multiput(310,60)(11,11){3}{\line(3,-1){60}} 
\multiput(322,56)(12,-4){5}{\line(1,1){58}} 
\multiput(250,90)(12,-4){6}{\line(0,-1){30}} 
\multiput(250,90)(12,-4){5}{\line(1,1){50}} 
\multiput(370,40)(11,11){3}{\line(0,-1){30}}
\multiput(408,78)(10,10){3}{\line(0,-1){30}}
\multiput(310,60)(0,-10){4}{\line(3,-1){60}} 
\multiput(370,40)(0,-10){4}{\line(1,1){58}}
\thicklines
\put(395,95){\vector(-3,1){95}}
\put(300,125){\vector(-1,-1){30}}
\put(270,97){\vector(3,-1){50}}
\put(320,78){\vector(0,-1){9}}
\put(320,71){\vector(3,-1){40}}
\put(358,59){\vector(1,1){36}}

\put(260,140){\vector(3,-1){20}} 
\put(260,140){\vector(0,-1){20}} 
\put(260,140){\vector(1,1){10}} 
\put(285,140){\makebox(0,0){$V$}}
\put(260,150){\makebox(0,0){$U$}}
\put(250,125){\makebox(0,0){$Z$}}
\end{picture}

Paths which circumnavigate the dislocation (pictured on the right) experience a jump. Intuitively, this extends the group of horizontal lattice translations generated by $U$ and  $V$ by vertical lattice translations $Z$ to give an integer Heisenberg group generated by unitaries $U$ and $V$, together with a central unitary $Z$, satisfying $VU = UVZ$.
The centre $\{Z^n: n\in \integer\}$ has irreducible representations $Z \mapsto e^{i\theta}$, so that we have an algebra bundle over $S^1$ with the fibre over $\theta \in S^1$ being the algebra with $VU = e^{i\theta}UV$, which, for almost all $\theta$, is an irrational rotation algebra. 
It is the T-dual of a continuous trace algebra defined using $P = \torus^3$ considered as a principal $\torus^2$-bundle over $S^1$ with the volume 3-form defining the $H$-flux.
This was a key observation in showing that the T-dual of a continuous trace algebra for the trivial $\torus^2$- bundle over a circle with the volume as $H$-flux, could be associated with a noncommutative geometry, \cite[\S 5]{MR}.


%
%

\section{Magnetic translations and monopoles}

We conclude with a more speculative suggestion. The magnetic flux through a surface surrounding a Dirac monopole would give rise to a 3-form, \cite{Di,WZ}, which is the local form of the group 3-cocycle discussed below.  Despite one early reported observation,  Dirac monopoles have proved elusive, but analogues have been found recently in spin ice, \cite{C08}, and in this section and the next we describe these in more detail.

In the presence of a magnetic field the momentum operators $P_j$ are replaced by $P_j-\frac{e}{c}A_j(X)$, where $e$ the electron charge, $c$ the speed of light, and $A$ is the magnetic vector potential. (Although we are now considering dimension $d=3$, we shall continue to use  $A$, $v$, $x$, $y$ etc. for 3-dimensional vectors.)
We can extend the natural translation operators $U_0(x)$ on $L^2(\real^3)$ to  magnetic translation operators $U_A(x)$, which acts on $\psi \in L^2(\real^3)$ by parallel transport:
$$
(U_A(x)\psi)(v) = \exp\left[-i\frac {e}{c\hbar}\int_0^1\, ds \, x.A(v + sx)\right] \psi( v + x).
$$
(Usually the magnetic translation operators are chosen with the opposite sign of $e$, since, in a suitable gauge,  the operators $P_j+\frac ec A_j$ commute with the kinetic energy, but the mathematics differs only by that sign change.) 
Comparing $(U_A(x)U_A(y)\psi)(v)$ and $(U_A(x+ y)\psi)(v)$ we see that the arguments of $\psi$ match, but the exponential factors in front have exponents differing  by the integral of $-i\frac {e}{c\hbar}A$ round the edges of the triangle $\Delta(v,x,y)$ with vertices $v$, $v + x$ and $v + x + y$. By Stokes' Theorem, this is the same as the integral of $-i\frac {e}{c\hbar}{dA}$ over the triangle, which is ($-ie/c\hbar$ times) the magnetic flux through the triangle.
The exponential defines a function $\phi_A(x,y)$ taking the wave function argument $v$  to the exponential of this (scaled) flux through 
$\Delta(v, x, y)$,  and the multiplier is the operator $\Phi_A(x,y)$ obtained by applying the function $\phi_A(x,y)$ to the position operator $X$, (symbolically the exponential of the flux through $\Delta(X,x,y)$, so that
$$
U_A(x)U_A(y) = \Phi_A(x,y)U_A(x+y).
$$
(Since they are all functions of the position observable $X$, the multipliers $\{\Phi_A(x,y)\}$ all commute with one another.) 
This result is given in \cite[Eq. (9)]{Z64}  (with the sign of $e$ changed to match our conventions), but the proof given there seems to work only for uniform magnetic fields, (which are the main focus of that paper). 

The operator nature of $\Phi_{\bf A}$ has some important consequences:
\begin{enumerate}
\item
As the multiplier is a function of $X$, it is consequently affected by translations, so that a magnetic translation through $x$ sends $\sigma_A$ to a translated version $x[\sigma_A]$;
\item
The products
\begin{eqnarray*}
\left(U_A(x)U_A(y)\right)U_A(z) &=& \Phi_A(x,y)U_A(x+y)U_A(z)\\
&=& \Phi_A(x,y)\Phi_A(x+y,z)U_A(x+y+z),\\
U_A(x)\left(U_A(y)U_A(z)\right) &=& U_A(x)\Phi_A(y,z)U_A(y+z)\\
&=& x[\Phi_A(y,z)]U_A(x)U_A(y+z)\\
&=& x[\Phi_A(y,z)]\Phi_A(x,y+z)U_A(x+y+z).
\end{eqnarray*}
differ by the factor $x[\Phi_A(y,z)]\Phi_A(x,y+z)/\Phi_A(x,y)\Phi_A(x+y,z)$, which is the exponential of the sum of fluxes through $\Delta(X+x,y,z)$ and $\Delta(X,x,y+z)$ minus those through $\Delta(X,x,y)$ and $\Delta(X,x+y,z)$. Together this gives the exponential of the flux out of the tetrahedron whose faces are these four triangles. 
In Maxwell theory without monopoles this flux must vanish and one has the cocycle identities
$$
\Phi_A(x,y)\Phi_A(x+y,z) = x[\Phi_A(y,z)]\Phi_A(x,y+z),
$$
and thus the associative law $U_A(x)\left(U_A(y)U_A(z)\right) = U_A(x)\left(U_A(y)U_A(z)\right)$.
However, any magnetic monopoles within a tetrahedron produce a net flux into or out of the tetrahedron, violating the cocycle identity so that the multipliers no longer cancel, and the operators $U_A(x,\xi)$ fail to be associative unless the flux obeys an integrality condition so that the exponential is 1. This integrality provides Dirac's argument for the quantisation of electric charge \cite{Di}. (At first sight the flux out of the tetrahedron should be the integral of $dF = d^2A = 0$ over the interior, but as noted in \cite{WZ}, this argument fails due to the Dirac-type string emanating from a monopole which must exit the tetrahedron through one of the surface triangles.)
In general  
$$
\Psi_A(x,y,z) = x[\Phi_A(y,z)]\Phi_A(x,y+z)/\Phi_A(x,y)\Phi_A(x+y,z)
$$
is  a group cohomological 3-cocycle \cite{WZ} controlling the associativity of the operators. This has been investigated in \cite{BHM06,BHM10} where it has been remarked that such nonassociativity is not serious, and can be handled by moving into an appropriate category of modules.
\end{enumerate}

\section{Magnetic monopoles in pyrochlore lattices}

In pyrochlore lattices such as Dy$_2$Ti$_2$O$_7$ and Ho$_2$Ti$_2$O$_7$ the rare earth atoms of dysprosium or holmium sit at the vertices of regular tetrahedra. In fact, each vertex at ${\bf v}$ is shared by two tetrahedra, with a point ${\bf p}$ of one corresponding to the point $2{\bf v}-{\bf p}$ of the other. Each tetrahedron can be labelled by its centroid ${\bf c}$ and these form a diamond lattice, (the 3-dimensional lattice describing carbon atoms in diamond). The tetrahedra can be considered positive or negative according to whether they are rotated or reflected versions of the some reference tetrahedron. The neighbours of each tetrahedron in the lattice have the opposite parity.
We write ${\bf e}_{{\bf v}}$ for the unit vector from ${\bf c}$ to ${\bf v}$.

\medskip
\begin{center}
\includegraphics[width=75mm]{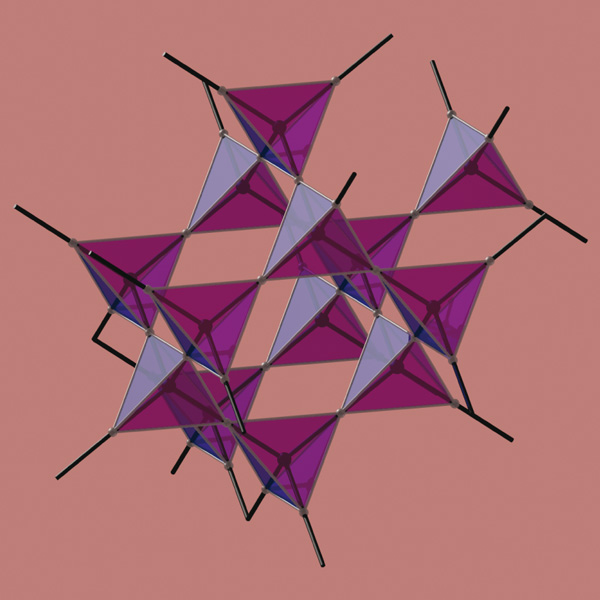}

Pyrochlore lattice with dipoles (taken from \cite{C08})
\end{center}

(We concentrate on the magnetic properties of the dysprosium and holmium compounds, but neutron scattering experiments suggest that other pyrochlore lattices such as Yb$_2$Ti$_2$O$_7$, Pr$_2$Zr$_2$O$_7$, and Tb$_2$Ti$_2$O$_7$  should display further interesting quantum effects. It seems that such quantum spin ice should belong to one of seven classes classified by their statistics and time reversal behaviour, \cite{WS16}, thus providing a link to the bulk--boundary features discussed earlier.)  

Electron spin provides a magnetic dipole at each vertex atom, which points along the line from the vertex to the centroid either inwards (towards the centroid) or outwards (away from it), so that the magnetic dipole has the form ${\bf S}_{{\bf v}} = S\epsilon_{{\bf v}}{\bf e}_{{\bf v}}$ with $\epsilon_{{\bf v}} = \pm 1$.
One usually refers to this as spin ice, as the lattice has similar geometry to Pauling's model of ordinary ice, with an oxygen atom at the centroid,  two hydrogen atoms inside each tetrahedron and two outside replacing two inward and two outward pointing dipoles, \cite{P35}.

The energy contribution of these dipoles can be described by a lattice dipole model:
$$
H = -J\sum_{{\bf v},{\bf v}^\prime} {\bf S}_{{\bf v}}\cdot {\bf S}_{{\bf v}^\prime} + Da^3\sum_{{\bf v},{\bf v}^\prime}
\left[\frac{{\bf S}_{{\bf v}}\cdot {\bf S}_{{\bf v}^\prime}}{|{\bf r}_{{\bf v}{\bf v}^\prime}|^3} - 3\frac{({\bf S}_{{\bf v}}\cdot {\bf r}_{{\bf v}{\bf v}^\prime}) ({\bf S}_{{\bf v}^\prime}\cdot {\bf r}_{{\bf v}{\bf v}^\prime})}{|{\bf r}_{{\bf v}{\bf v}^\prime}|^5}\right],
$$
where ${\bf r}_{{\bf v}{\bf v}^\prime} = ({\bf v}-{\bf v}^\prime)$, $J$ and $D$ are constants, and $a$ denotes the distance between nearest neighbours. (There is always some tetrahedron containing both the nearest neighbours ${\bf v}$ and ${\bf v}^\prime$, and so $a$ is the length of a tetrahedron edge.)  

Although two-dimensional lattice models of this kind were solved by Lieb, Sutherland, and Baxter, \cite[Ch. 8,9]{B82}, three-dimensional lattices are much less tractable. 
We therefore restrict our attention to nearest neighbour interactions, and note that the contribution of this tetrahedron to the first term is 
$$
-JS^2\sum_{{\bf v},{\bf v}^\prime} \epsilon_{\bf v} \epsilon_{{\bf v}^\prime}\,{\bf e}_{{\bf v}}\cdot {\bf e}_{{\bf v}^\prime}.
$$
The contribution of ${\bf v}^\prime ={\bf v}$ is just $-JS^2$, but, otherwise, ${\bf e}_{{\bf v}}\cdot {\bf e}_{{\bf v}^\prime} = -1/3$
and one has the contribution
$$
\frac{JS^2}3 \sum_{{\bf v}\neq {\bf v}^\prime} \epsilon_{{\bf v}} \epsilon_{{\bf v}^\prime}
= \frac{JS^2}3 \left(\left|\sum_{{\bf v}} \epsilon_{{\bf v}}\right|^2  - 4\right).
$$
In this form it is clear that this energy contribution is minimum when $\sum_{{\bf v}} \epsilon_{{\bf v}} = 0$, that is, when two dipoles point inwards and two outwards.  (The second, more complicated, term in $H$ also reduces, and gives the same conclusion, when only nearest neighbours are considered.)
The ground state of the system (for nearest neighbour interactions) thus occurs when two dipoles point into and two out of each tetrahedron.
 
It is useful to visualise the dipoles in terms of the \lq\lq dumbell model\rq\rq\, in which each dipole is considered as having its north pole at the centroid of the tetrahedron into which it points, and its south pole at the centroid of the tetrahedron away from which it points.
The balance between inward- and outward-pointing dipoles mean that each centroid has two north and two south poles, and consequently no magnetic \lq\lq charge\rq\rq.

It is possible, however, to flip a dipole at one vertex to produce an excited state in which the two  tetrahedra meeting at that vertex have an imbalance with an extra dipole pointing inwards, and the other with an extra dipole pointing outwards. In the dumbell model  one tetrahedron has  three north poles and one south, 
and its neighbour has three south poles and one north, so that one seems to contain a north monopole, and the other a south.
Using our previous expressions the energy of this excited state is $\frac{8JS^2}3$ above the ground state (plus a rapidly decreasing function of interatomic distances from the second term in $H$.)

Flipping another dipole in one of these tetrahedra will restore the balance in that one, but create an imbalance in another neighbour. Overall there are still two tetrahedra with an imbalance but they are no longer neighbours, and the energy is the almost the same as before. This process can be iterated to move the monopoles as far apart as one wishes (even to infinity), with little energy needed, so that the monopoles are essentially free. 
The predicted consequences have been observed experimentally, \cite{C08},  see also \cite{BG02,JH10}.

Putting the pieces together we now have tetrahedra containing magnetic monopoles, and, by the magnetic translation analysis, we know that this leads to the presence of a group 3-cocycle.
However, that cocycle was given by a non-zero flux through a tetrahedron.
It is the magnetic field ${\bf H}$ rather than ${\bf B}$ which gives the interesting flux, since the divergence of ${\bf B}$ still vanishes in the pyrochlore lattice, (as there are no Dirac monopoles), but ${\rm div}({\bf H})$  does not. Its Hodge star is a non-zero 3-form $*{\rm div}({\bf H})$, which  is exact in $\real^3$.
(There is a constraint in T-duality as the cohomology class of $*{\rm div}({\bf H})$ should be integral, but this differs from the Dirac integrality condition.)
With strict lattice translation symmetry the 3-form could be lifted from a closed but not exact 3-form on the quotient of $\real^3$ by the diamond lattice.
There is a tension between the desire to obtain a non-trivial cohomology class and the fact that the monopole breaks strict lattice translation symmetry, but perhaps this flux will provide a useful approximation in suitably engineered situations, just as monopoles provide a useful way of intepreting flipped magnetic dipoles. 
Alternatively, one might be appropriate to switch to compactly  supported cohomology, as discussed in  \cite[\S 2]{MW}. 


\section*{Acknowledgements}

This paper is based on a talk at the IGA/AMSI workshop on {\it Topological matter, strings, K-theory and related areas} in Adelaide in September 2016, and I'm pleased to acknowledge support from ARC Discovery Project grant DP150100008. 
Some later stages of the writing were carried out  in Cambridge, supported by EPSRC grant number EP/K032208/1. The author would like to thank the Isaac Newton Institute for Mathematical Sciences for support and hospitality during the {\it Operator algebras: subfactors and applications} programme when that work was undertaken.

Many of the ideas arose out of discussions with V.\  Mathai and G.C.\ Thiang over several years.
I'm grateful to Stephen Blundell in the Oxford Condensed Matter group for very useful conversations about spin ice, and for drawing to my attention to several  papers which I would otherwise have missed.
Finally thanks to a very careful referee who made numerous useful suggestions.


\appendix

%
%

\section{Noncommutative geometry and K-theory}

The operators representing quantum-mechanical observables can be added,  multiplied by complex scalars, and by each other, and their adjoints formed, and they can conveniently be modelled by C$^*$-algebras (which are norm-closed *-subalgebras of the bounded operators on a Hilbert space).
Noncommutative geometry is motivated by the Gel'fand--Naimark characterisation of commutative C$^*$-algebras: 

\begin{theorem} 
Every commutative C$^*$-algebra $\alg$ is isometrically *-isomorphic to $C_0(M)$, with pointwise  multiplication, star operation, and the supremum norm $\|f\| = \sup\{ |f(m)|: m\in M\}$, for some locally compact Hausdorff space $M$, which can be identified with the spectrum of $\alg$. Compact spaces correspond to unital commutative C$^*$-algebras.
\end{theorem}

A continuous map  $\phi: M_1 \to M_2$ between two compact spaces induces a  *-homomorphism $\phi^*:C(M_2) \to C(M_1)$ given by $\phi^*(f) = f\circ\phi$, so that the category of commutative unital C$^*$ algebras with  *-homomorphisms is contravariantly equivalent to the category of compact Hausdorff spaces and continuous maps, and this extends with minor modifications to locally compact spaces and non-unital algebras. This observation enables us to replace geometry by commutative algebra. 
This captures even some quite subtle geometric features, for example, the Serre--Swan Theorem gives an algebraic characterisation of vector bundles on $M$:

\begin{theorem}
The space of sections $\Gamma(E)$ of a finite rank vector bundle $E$ over a compact space $M$ has a pointwise multiplication action of $C(M)$,  with respect to which it is a finite rank projective $C(M)$-module, (that is,the image under a projection $p$ of a finite rank free $C(M)$-module). Conversely,  every finite rank projective $C(M)$-module such module arises as sections $\Gamma(E)$ of a finite rank vector bundle $E$ over $M$.  
\end{theorem}

The K-theory of a manifold is defined using equivalence classes of vector bundles, so the Serre-Swan correspondence allows one to interpret the group $K^0(M)$ of equivalence classes of vector bundles, with a group $K_0(C(M))$ of appropriate equivalence classes of projections. 
(One subtlety is that these are projections on free modules $C(M)^n$, so we need in the matrix algebras $M_n(\alg)$ with entries in $\alg = C(M)$. 
Several equivalence relations  can be used, for instance,  homotopy equivalence, unitary equivalence,  and  von Neumann equivalence ($p = u^*u \sim q= uu^*$), amongst others, all give the same K-theory.)
A homomorphism $\varphi: \alg_1 \to \alg_2$ will also define a homomorphism $M_n(\alg_1) \to M_n(\alg_2)$ which sends projections to projections. This is compatible with the equivalence relation and so gives a homomorphism $\varphi_*: K_0(\alg_1) \to K_0(\alg_2)$.

In geometry one then defines a group $K^1(M)$ as $K^0(SM)$, where $SM$ denotes the suspension of $M$.
The algebraic analogue of the suspension is to replace an algebra $\alg$ by $C_0(\real,\alg)$ (equipped with pointwise multiplication and star operation).
So we set $K_1(\alg) = K_0\left(C_0(\real,\alg)\right)$. One can also define $K_1(\alg)$ using unitary elements of the matrix algebras. 
Bott's Periodicity Theorem tells us that it is pointless trying to define further topological K-groups by more suspensions, since $K^0(S^2M)$  is isomorphic to $K^0(M)$. The same applies in the case of complex C$^*$-algebras, so that there is no need to look at $K_j(\alg)$ for $j>2$.

The  C$^*$-algebras appearing in quantum systems are generally noncommutative, but  their spectra do not always define a useful  geometry. 
The algebraic interpretation in the above discussion opens the way to replacing $C_0(M)$ by a general noncommutative C$^*$-algebra $\alg$, and so obtain a noncommutative geometry with well-defined  $K$-theory.

Two noncommutative examples which will play an important part in this paper later are the following:
\begin{enumerate}
\item
The algebra of compact operators, $\cpt(\hilb)$, generated by finite rank operators (or even by rank 1 projections.
\item
The Toeplitz algebra $\cT$  generated by the shift operator $S$ in $\ell^2({\mathbb N})$ is noncommutative since $S^*S \neq SS^*$.
\end{enumerate}

%
%

\section{Morita equivalence and continuous trace algebras}

Representations of C$^*$-algebras, whether the irreducibles which give the spectrum, or the projective modules which define the K-theory,  perform such  an important role in the theory that it is helpful to know when two algebras have the same representation theory. We shall say that 
two algebras
$\alg_1$ and $\alg_2$ are Morita equivalent when there is a natural equivalence between the categories of $\alg_1$-modules and $\alg_2$-modules. 
(Since the K-groups $K_j(\alg)$ are defined  in terms of $\alg$-modules, Morita equivalent algebras have the same K-theory.)

For example the compact operators $\cpt(\hilb)$ have only one irreducible representation (the obvious one on $\hilb$), and they are all Morita equivalent to each other, and, in particular, to $\cpt(\complex) \cong \complex$. 
More generally, the algebra $C_0(P,\cpt(\hilb))$ is Morita equivalent to $C_0(P)$.
The latter example is a special case of another useful class of C$^*$-algebras, the continuous trace algebras, which are only locally Morita equivalent to continuous functions on the spectrum \cite{RW}: 

\begin{definition}
Let $\{P_j\}$ be an open covering of the spectrum $P = \spec(\alg)$. If $\alg^{P_j}$ (the part of $\alg$ whose spectrum lies in $P_j$) is Morita equivalent to $C_0(P_j)$ then $\alg$ is said to be a {\em continuous trace algebra}. 
\end{definition}

Commutative C$^*$-algebras are certainly continuous trace algebras, and so are algebras of the form 
$C_0(P,\cpt(\hilb))$ . More general continuous trace algebras can be thought of as sections of a compact operator algebra bundle over the spectrum, and they have an interesting characterisation.
To piece together what happens in an overlap $P_{jk} = P_j\cap P_k$ we must find a transition automorphism from the compact operators over $P_j$ to those over $P_k$, and the only automorphisms of $\cpt(\hilb)$ are those given by $A \mapsto UAU^{-1}$ for some unitary operator $U$ on $\hilb$. So we pick a continuous map $U_{jk}:P_{jk} \to \cU(\hilb)$ to implement the transition automorphism over $P_j\cap P_k$. Next we consider triple intersections $P_{jk\ell}$ where $U_{jk}$, $U_{k\ell}$, and  $U_{j\ell}$ are all defined. Since $U_{jk}U_{k\ell}$, and  $U_{j\ell}$ both implement the transition between $P_j$ and $P_\ell$ they must define  the same automorphism, and this means that they differ only by a scalar function $c_{jk\ell}\in C_b(P_{jkl})$. By unitarity $c_{jk\ell}$ has modulus 1, and it is easy to see that it defines a \v{C}ech cocycle in $H^2(P,\torus)$, which, by a standard argument is isomorphic to 
$H^3(P, \integer)$, and we denote by $\delta$ the cohomology class defined by $c_{jk\ell}$.

\begin{theorem} {\bf (Dixmier--Douady, \cite{DD63}.)}
For every locally compact space  $P$ and $\delta\in H^3(P,\integer)$ there is a continuous trace algebra $\alg = CT(P,\delta)$  with spectrum $P$ and cohomology class  $\delta$, and this is unique up to Morita equivalence.
\end{theorem}

Thus the continuous trace algebras provide a wider class of C$^*$-algebras than the commutative algebras, whilst retaining a strong geometrical flavour. 
In string theory the Dixmier--Douady class may be thought of as the H-flux through $P$.  
The flux is automatically trivial in dimensions less than three and so it could be relevant for the bulk but not the boundary of three-dimensional objects. 

So far we have talked about complex C$^*$-algebras, but antilinear symmetries of a physical systems, such as  time inversion, preserve only the real structure and  it is more appropriate to work with real C$^*$-algebras (closed *-subalgebras of the bounded operators on a real Hilbert space).
In that case we define real K-groups $K\!R_{p,q}(\alg) = K_0(\alg\otimes {\rm Cliff}(\real^{p,q}))$ where ${\rm Cliff}(\real^{p,q})$ denotes the Clifford algebra of a pseudo-Euclidean space with quadratic form of signature $(p,q)$, and one has to use a graded theory, \cite{ABS,K,KK88}. The K-groups depend only on the difference $p-q$ and these have periodicity 8, \cite{A66}.
The  classification of quantum systems into 10 = 2 + 8 classes suggested a possible link to K-theory and  Bott periodicities.




\end{document}